\documentclass[11pt]{article}

\usepackage{amsmath}
\usepackage{amssymb}
\usepackage{amsthm}
\usepackage{mathtools}
\usepackage[colorlinks=true, linkcolor=red, citecolor=blue, urlcolor=blue]{hyperref}
\usepackage{bm}
\usepackage{algorithm}
\usepackage[noend]{algpseudocode}
\usepackage[capitalize]{cleveref}
\usepackage{mdframed}
\usepackage{aliascnt}

\providecommand{\email}[1]{\href{mailto:#1}{\nolinkurl{#1}}}

\newtheorem{theorem}{Theorem}[section]

\newaliascnt{lemma}{theorem}
\newtheorem{lemma}[lemma]{Lemma}
\aliascntresetthe{lemma}

\newaliascnt{corollary}{theorem}
\newtheorem{corollary}[corollary]{Corollary}
\aliascntresetthe{corollary}

\newaliascnt{remark}{theorem}
\newtheorem{remark}[remark]{Remark}
\aliascntresetthe{remark}

\newaliascnt{definition}{theorem}
\newtheorem{definition}[definition]{Definition}
\aliascntresetthe{definition}

\crefname{theorem}{Theorem}{Theorems}
\crefname{lemma}{Lemma}{Lemmas}
\crefname{corollary}{Corollary}{Corollaries}
\crefname{remark}{Remark}{Remarks}
\crefname{definition}{Definition}{Definitions}

\newcommand{\mean}{\bm{\mu}}
\newcommand{\eps}{\varepsilon}
\newcommand{\cY}{\mathcal{Y}}
\newcommand{\tilded}{L}

\usepackage{fullpage,palatino}

\title{Distributed Gaussian Mean Testing under Communication Constraints: messages, samples, and coins}

\author{Clément L. Canonne\thanks{University of Sydney. Email: \email{clement.canonne@sydney.edu.au}.} \and Nimitt\thanks{IIT Gandhinagar. Email: \email{nimittnim@gmail.com}.}}

\begin{document}

\maketitle

\begin{abstract}
 We revisit the problem of Gaussian mean testing in a distributed, communication constrained setting, where each of $n$ users independently observes samples from an unknown $d$-dimensional spherical Gaussian distribution $\mathcal{G}(\mean,\mathbb{I}_d)$, and can communicate up to $\ell$ bits to a central referee. The referee's goal is then to distinguish between cases (i)~$\|\mean\|_2 = 0$ versus (ii)~$\|\mean\|_2\ge \eps$. This problem has been considered in the private- and public-coin settings, when each user holds exactly \emph{one} sample, or more generally when each holds exactly $m$ samples. In this work, we significantly generalize the question in three directions: when the users only share a small number $s$ of random bits, when each user holds a different number of samples $m_k$, and when each user can send a different number of bits $\ell_k$ to the referee.

\end{abstract}

\section{Introduction}
You have been put in charge of a new program, to monitor underwater volcanic activity in a large sector of the ocean. To maximize coverage, you have identified a set of $n$ locations where to deploy buoys equipped with high-tech measuring devices, strategically chosen to maximize coverage. Each of these sensors will measure the volcanic activity at regular intervals, and send hourly the result to the central station for processing and analysis, to detect any risk. But choosing the location of these buoys was only the first step: the ocean is a harsh and big place, and as a result these sensors are not only quite distant from the station, they also have strong limitations on how many good-quality measurements they can make and how much data they can accurately transmit before running out of energy for the day. The measurements they make are subject to significant background noise; and even broadcasting data \emph{to} the network of buoys to synchronize them is complex, given how remote they all are. So you are left with a daunting, highly constrained task: 
\begin{mdframed}\itshape
    How to design a good detection protocol under heterogeneous constraints, where each sensor can make a different number of measurements, send a different amount of data back to the center, and they all only share a limited amount of coordination?
\end{mdframed}

\paragraph{Problem Statement.} We formulate the problem as a composite hypothesis testing task. Let $\mathcal{G}(\mean,\mathbb{I}_d)$ denote a $d$-dimensional Gaussian distribution with mean vector $\mean \in \mathbb{R}^d$ and covariance $\mathbb{I}_d$. Define the following hypothesis classes for a pre-specified threshold $\eps \in (0,1]$:
\[
    \mathcal{H}_0:= \{\mathcal{G}(\bm{0},\mathbb{I}_d)\} \quad \mathcal{H}_1:= \{\mathcal{G}(\mean,\mathbb{I}_d): \|\mean\|_2\ge \eps\}
\]
In a distributed setting, $n$ users observe i.i.d.\ samples from an unknown distribution $G \in\mathcal{H}_0 \cup \mathcal{H}_1$ and transmit $\ell$-bit messages to a central referee. Specifically, user $k \in[n]$ observes $m_k \ge 1$ independent samples from $G$:
\[
X^{(k)}= \left(\bm{X}^{(k,1)},\bm{X}^{(k,2)},\cdots, \bm{X}^{(k,m_k)}\right), \quad \bm{X}^{(k,i)} \sim G.
\]
and applies a (possibly randomized) encoding function $W^{(k)}\colon \mathbb{R}^{d\cdot m_k} \to \cY_k := \{0,1\}^{\ell_k}$
to it, before sending the resulting message $\bm{Y}^{(k)} = W^{(k)}\left(X^{(k)}\right)$ to the referee. %
The referee then applies a test $T\colon \cY_1\times\cdots\times\cY_n \to \{0,1\}$ to the vector of received messages
$Y^{n} =\left(\bm{Y}^{(1)},\bm{Y}^{(2)},\cdots,\bm{Y}^{(n)}\right)$, and accepts $\mathcal{H}_0$ if, and only if, $T(Y^{n} ) = 0$.

We allow the users to share $s \geq 0$ bits of \emph{public randomness}, which are used to jointly sample the tuple encoding functions $W^n = (W^{(1)},W^{(2)},\cdots,W^{(n)})$.\footnote{That is, the encoding functions (or ``channels'') are randomized using both $s\geq 0$ bits of public randomness, and an arbitrary number of bits of private randomness (independent across users).} We define the \emph{two-sided error} for protocol $\Pi := (W^n,T)$ as
\[
    \delta(\Pi) = \max \left(\sup_{G \in \mathcal{H}_0} \Pr\left[T(Y^n)=1\right], \sup_{G\in\mathcal{H}_1}\Pr \left[T(Y^n) =0\right] \right)  
\]
\begin{figure}[H]
    \centering
    \includegraphics[width=0.70\textwidth]{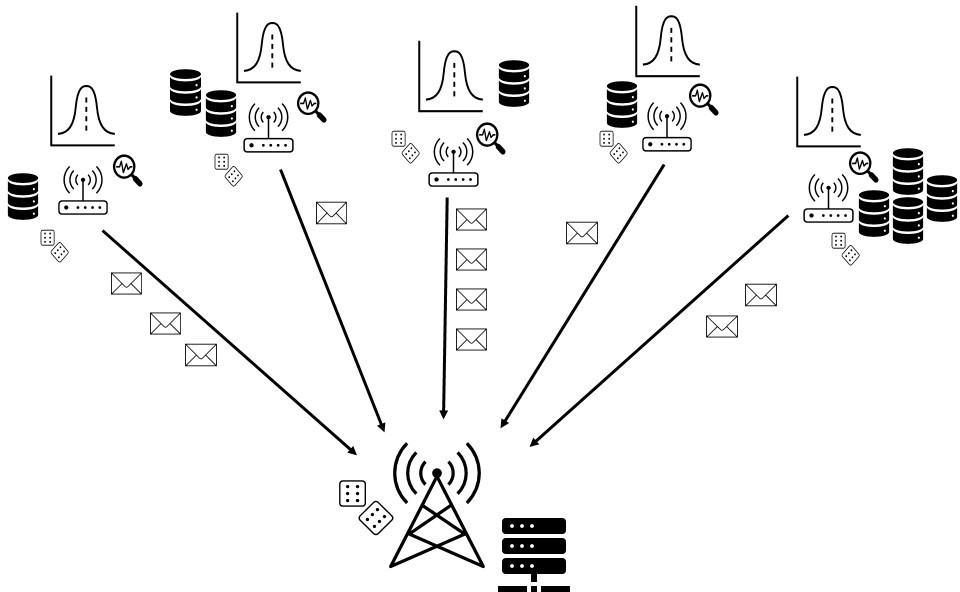}
    \caption{A depiction of the setting considered in this work: several distributed machines (``users'') are measuring the same signal (vector $\mean\in\mathbb{R}^d$) subject to white noise (spherical Gaussian). They all share a short random seed, and have different capabilities (amount of data collected, amount of communication they can send): by sending a message to a central server (``referee''), they must enable it to test whether the underlying signal measure exists ($\|\mean\|_2> \eps$) or is pure noise ($\mean = \bm{0}$).}
    \label{fig:setting}
\end{figure}
The goal is to find required amount of \emph{resources} (amount of shared randomness $s$, numbers of samples $(m_k)_k$, individual communication budgets $(\ell_k)_k$) such that there exists a protocol $\Pi$ with two-sided error bounded by a target error probability $\delta$ (unless specified otherwise, fixed as a small constant  $1/10$). For simplicity, we decompose the problem into following three subproblems: 

\begin{enumerate}
    \item \textbf{Limited public coin}: Each user observes only one sample, and sends $\ell$ bits ($m_k=1$, $\ell_k=\ell$ for all $k$); the users are allowed to share $s$ bits of public randomness. What is the \emph{sample complexity}, that is, the minimum $n=n(d,s,\ell,\eps)$ such that $\inf_{\Pi} \delta(\Pi) \le 1/10$?
    \item \textbf{Heterogeneous samples}: The users share $s = \Omega(\log(d/\ell))$ bits of public randomness, and each sends $\ell$ bits ($\ell_k=\ell$ for all $k$). What are sufficient conditions on $\bm{m}=(m_1,\dots,m_n)$ such that there exists a protocol $\Pi$ with $\delta(\Pi) \le 1/10$?
    \item \textbf{Heterogeneous communication}: The users share $s = \Omega(\log(d/\ell))$ bits of public randomness, and each observes one sample ($m_k=1$ for all $k$). What are sufficient conditions on $\bm{\ell}=(\ell_1,\dots,\ell_n)$ such that there exists a protocol $\Pi$ with $\delta(\Pi) \le 1/10$?
\end{enumerate}

\paragraph{Our results.} We now summarize our results for these three subproblems. 

\begin{theorem}[Limited public-coin]\label{sdgmt} For $\ell \in [d]$ and $s\ge 0$, there exists an $s$-public coin protocol for distributed Gaussian mean testing under $\ell$-bit communication constraint with 
    \[O\left(\frac{d}{\ell\eps^2}\cdot\sqrt{\max\left(\frac{d}{2^{ \Theta(s) }} , \ell\right)}\right)
    \]
    users, each holding one sample.
\end{theorem}
Note that for $s = 0$ and $s = O(\log(d/\ell))$, this matches the \emph{private-} and \emph{public-coin}  bounds of $O(d^{3/2}/(\ell\eps^2))$ and $O(d/(\sqrt{\ell}\eps^2))$ from \cite{AcharyaCT20}, known to be tight~\cite{SzaboVZ23}.

\begin{theorem}[Heterogeneous samples] \label{msampledgmt}
    For $\ell \in [r,rd]$, $s= \Omega(\log(d/\ell))$, and $\bm{m} = (m_1,\dots, m_n) \in \mathbb{Z}_+^n$ satisfying $m_k \in [r,r\eps^{-2}]$ for all $k \in [n]$, there exists a protocol for distributed Gaussian mean testing under $\ell$-bit communication constraint with $n$ users provided that 
    \[
    \frac{\sum_{1\leq k_1 \neq k_2 \leq n} \sqrt{m_{k_1}{m_{k_2}}}}{n}  \geq C\cdot\frac{d}{\sqrt\ell\eps^2} 
    \]
    where $C,r>0$ are absolute constants.
\end{theorem}
To interpret this result, observe that, at one extreme, with $m_k=\Theta(1)$ for all $k$, the expression matches the \emph{public-coin} bound. Moreover, we note that for a fixed ``budget'' of samples $m = \sum_{k=1}^n m_k$, a ``balanced'' allocation of samples across users maximizes the left-hand side, and hence is better (at least for our testing protocol).

\begin{theorem}[Heterogeneous communication] \label{combudgetsdgmt}
    For $\bm{\ell} = (\ell_1,\dots, \ell_n) \in \mathbb{Z}_+^n$ and $s= \Omega(\log(d/\|\bm{\ell}\|_\infty))$, there exists a protocol for distributed Gaussian mean testing under $\bm{\ell}$-bit communication constraint with $n$ users, where user $k$ has one sample and can send $\ell_k$ bits to the referee, provided that
    \[
    \frac{\|\bm{\ell}\|_1}{\sqrt{\|\bm{\ell}\|_\infty}}  \geq C\cdot\frac{d}{\eps^2} 
    \]
    where $C>0$ is an absolute constant.
\end{theorem}
Letting $\tilde{\ell} = \frac{1}{n}\sum_{k=1}^n \ell_k$ be the average communication budget, one can rephrase the sufficient condition as
    \[
    n\cdot \frac{\tilde{\ell}}{\sqrt{\max_{1\leq k\leq n} \ell_k}}  \geq C\cdot\frac{d}{\eps^2} 
    \]
which, for a fixed value of $\tilde{\ell}$, is easiest to satisfy when all communication constraints are the same (``balanced'').\medskip

We emphasize that this separation in three subproblems is for clarity and simplicity of exposition: we also provide a protocol involving a ``mix-and-match'' of all parameters, where public randomness $s$, sample counts $\bm{m} = (m_1,\dots, m_n)$, and individual communication budgets $\bm{\ell} = (\ell_1,\dots, \ell_n)$ are all kept as free parameters. The guarantees of this protocol are summarized in the following theorem:
\begin{theorem}[Everything, everywhere, all at once]\label{theorem:mix-and-match}
     For $s\geq 0$, $\bm{m} = (m_1,\dots, m_n) \in \mathbb{Z}_+^n$, and $\bm{\ell} = (\ell_1,\dots, \ell_n) \in \mathbb{Z}_+^n$, there exists a protocol for distributed Gaussian mean testing with $n$ users sharing $s$ bits of public randomness, where user $k$ has $m_k$ samples and can send $\ell_k$ bits to the referee, provided that there 
    exists a partition $\mathcal{P} = \{P_1,\dots,P_K\}$ of the $n$ users such that
    \[
    \frac{ \sum_{j_1\neq j_2 \in [K]}\sqrt{\min_{i_1\in P_{j_1}}(m_{i_1}) \cdot \min_{i_2\in P_{j_2}}(m_{i_2})}}{K} \ge C \cdot \frac{d}{\eps^2\sqrt{\max\left(\frac{d}{2^{\Theta(s)}} ,\|\bm{\ell}\|_\infty \right)}}
    \]
    and, for every $P \in \mathcal{P}$, $\sum_{i\in P}\ell_i \ge r\cdot\max\left(\frac{d}{2^{\Theta(s)}} ,\|\bm{\ell}\|_\infty \right)$ and $\min_{i \in P}(m_i) \in [r,r\eps^{-2}]$;
    where $C,r>0$ are absolute constants.
\end{theorem}

\paragraph{Related work.}
In the centralized setting where all samples are fully available to the referee, the tight sample complexity of for Gaussian mean testing is known, and, while considered ``folklore'', can for instance be found in \cite{inbook,DiakonikolasKP23} (as a special case of more general problem). Subsequent work has expanded on this foundation, exploring Gaussian mean testing under communication constraints \cite{AcharyaCT20} and truncation models \cite{CanonneGWY25}. \cite{AcharyaCT20} studied the problem in two variants of homogeneous regime (one sample and same communication budget $\ell$ for all users), namely the \emph{private-coin} ($s= 0$) and \emph{public-coin} ($s=\infty$), and showed sample complexities of
\[
O\left(\frac{d^{3/2}}{\ell\eps^2}\right) \quad (\emph{private coin}), \ O\left( \frac{d}{\sqrt{\ell}\eps^2}\right) \quad (\emph{public coin})
\]
While the authors only established upper bounds and (limited) lower bounds, follow-up work~\cite{SzaboVZ22,SzaboVZ23} showed these two sample complexities to be tight. (Moreover, exploiting standard properties of the Gaussian distribution, the above easily generalizes to the case where each of the $n$ users holds the same number of samples $m\geq 1$.\footnote{Namely, if $\bm{X}^{(k,1)},\bm{X}^{(k,2)},\cdots, \bm{X}^{(k,m_k)}\sim\mathcal{G}(\mean,\mathbb{I}_d)$, then \smash{$\tilde{\bm{X}}^{(k)} := \frac{1}{\sqrt{m}}\sum_{i=1}^m\bm{X}^{(k,i)} \sim \mathcal{G}(\sqrt{m}\mean,\mathbb{I}_d)$}, so one can use the one-sample protocol with new distance parameter $\eps' = \sqrt{m}\eps$ and ``single sample per user'' $\tilde{\bm{X}}^{(k)}$.\label{ft:aggregate}})

The above works reveal a $\sqrt{d/\ell}$ gap between the private- and public-coin sample complexities, highlighting the value of shared randomness. However, the public-coin protocol of \cite{AcharyaCT20} (as well as the similar protocol of~\cite{SzaboVZ23}) rely on jointly sampling a uniformly random rotation of the space, which technically requires infinitely many shared random bits~--~even with adequate discretization, the amount of shared randomness necessary is still impractical. One can thus wonder what can be achieved with a \emph{smaller} ``shared randomness budget:'' similarly, what happens if each user holds a different number of samples $m_k$, or can communicate a different number of bits $\ell_k$. To the best of our knowledge, all of these remain essentially understudied and open in the context of Gaussian mean testing, and even distribution testing at large, with the two exceptions outlined below. 

On the ``limited randomness'' side, our techniques and results for the $s$-public coin setting are reminiscent of the work of \cite{AcharyaCHST20}, which study the problem of \emph{identity testing for discrete distributions} under communication constraints. They provide an optimal protocol for identity testing of distributions over a finite domain $[k]$ under $\ell$-bit communication constraint with 
\[
    n = \Theta\left(\frac{\sqrt{k}}{\eps^2}\sqrt{\frac{d}{2^\ell}\vee1}\sqrt{\frac{k}{2^{\ell+s}}\vee1}\right)
\] 
users, giving clear interpolation between \emph{private-} and \emph{public-coin}. Our Blockwise Randomized Hadamard Transform (BRHT) primitive (see~\cref{sec:brht}) for compression is similar to their technique ``Domain Compression'', showing an exponential saving (in the number of shared random bits) in the sample complexity of the problem. Interestingly, however, both the actual techniques used, and the resulting tradeoffs, are quite incomparable: the upshot of \cite{AcharyaCHST20} was that ``one bit of communication is equivalent to two bits of shared randomness'' (in terms of sample complexity), which does no longer hold in the Gaussian mean testing setting.

On the ``heterogeneous communication'' side, while we are not aware of previous work for the case of Gaussian mean testing, we note that in the case of private-coin protocols for testing of discrete distributions, a baseline upper bound is hinted at in~\cite[Exercise~4.4]{Canonne22}, while \emph{lower bounds} for \emph{estimation} (not testing) are derived in~\cite{AcharyaCST23}. 

We conclude by mentioning that, in the (as discussed above, analogous, yet orthogonal) setting of identity testing of discrete distributions, under local privacy constraints, the recent work~\cite{CanonneGS26} states that \emph{``it is not at all obvious how this [heterogeneous sample sizes] could be handled neatly, and general results for this model could greatly
help practical implementations.''} While our results do not directly address their question (as we deal with mean testing and not discrete distributions, and communication instead of local privacy constraints), we hope that some of our ideas can help in this regard, especially given the observed similarities between communication and local privacy constraints~\cite{AcharyaCST23}.

\section{Preliminaries} 
\paragraph{Notation.} Throughout, we write $[n] := \{1,2,\dots,n\}$, and denote by $\bm{u}_d := (1/2, 1/{2},\dots, 1/2) \in \mathbb{R}^d$ the mean vector of a $d$-dimensional binary product uniform distribution over $\{0,1\}^d$, by $\mathbb{I}_d$ the $d$-dimensional identity matrix. We use the standard asymptotic notations $O(\cdot),\Omega(\cdot)$ and $\Theta(\cdot)$.

\paragraph{Simplification.} For clarity of exposition, we assume throughout that $d = 2^{k_1}$ and $\ell = 2^{k_2}$ for some integers $k_1,k_2\ge0$. This assumption is without loss of generality up to constant factors. Indeed, for arbitrary $d$, one can  embed into dimension $d'$ with \[d\le d'<2d, \qquad d' = 2^{k_1}\] via zero-padding. Similarly, for arbitrary $\ell$, one can use $\ell'$ bits with \[\ell/2<\ell'\le\ell, \qquad \ell' = 2^{k_2}.\] 

\noindent We will heavily rely on the following tools and techniques from the literature, which we recall below.

\begin{definition}[Normalized Sylvester Hadamard Matrix]
The \emph{normalized Sylvester Hadamard matrix} $H_d \in \mathbb{R}^{d\times d}$ is recursively defined by
\[
    H_1 := [1],
\]
and for every integer $d \ge 1$
\[
    H_{2d} := \frac{1}{\sqrt{2}}\begin{bmatrix}
        H_d & H_d \\ H_d & -H_d
    \end{bmatrix}
\]  
\end{definition}

\begin{lemma}[Paley--Zygmund inequality]\label{lemma:paley-zygmund} For a non-negative random variable $Z$ and $0\le\theta\le1$
\[
    \Pr\!\left[Z > \theta\mathbb{E}[Z] \right] \ge (1-\theta)^2\frac{\mathbb{E}[Z]^2}{\mathbb{E}[Z^2]}
\]
\end{lemma}

\paragraph{Binary product mean testing.} Let $P$ be an unknown $d$-dimensional binary product distribution  over $\{0,1\}^d$ with mean vector $\bm{p} \in (0,1)^d$. Given i.i.d.\ samples from $P$ and a threshold $\eps > 0$, the binary product mean testing problem is to distinguish between cases $\bm{p} = \bm{u}_d$ versus $\|\bm{p} - \bm{u}_d\|_2\ge \eps$ with probability at least $99/100$. It is known that the optimal sample complexity for this problem is $\Theta(\sqrt{d}/\eps^2)$, and can be achieved by~\cref{uniformity_testing} below. For completeness, we provide the analysis of~\cref{uniformity_testing} in Appendix, ~\cref{proof_uniformity_testing}.  

\begin{algorithm} 
        \caption{Binary product mean testing}
        \label{uniformity_testing}
        \begin{algorithmic}[1]
            \Require Samples $\bm{X}^{(1)},\dots,\bm{X}^{(n)} \in \{0,1\}^d$; threshold $\eps > 0$. 
            \State $\tau \gets \eps^2/2 $ 
            \For{$i \in [d]$}
            \State $T_i \gets \sum_{k_1 \neq k_2 \in [n]}(X_i^{(k_1)}-1/2)(X_i^{(k_2)}-1/2)$ 
            \EndFor
            \State $T \gets \frac{1}{n(n-1)}\sum_{i=1}^d T_i$
            \If {$T > \tau$}
            \State \Return $\textsf{reject}$ 
            \Else
            \State \Return $\textsf{accept}$ 
            \EndIf
        \end{algorithmic}
\end{algorithm}

\paragraph{Private-coin.} As discussed above, the private-coin setting, which corresponds to $s=0$ bits of shared randomness, was studied for Gaussian mean testing in \cite{AcharyaCT20}. We briefly describe their protocol, as we will use it later on, which uses two main ingredients: a reduction to the binary product case, and a simple private-coin algorithm to solve the latter.
\begin{lemma}[Reduction to binary product mean testing] \label{lemma:reduction_to_bpmt}
Let $\eps\in(0,1)$. If there exists an algorithm for a binary product mean testing with $n(d,\eps)$ samples, there exists an algorithm for Gaussian mean testing with $n(d,\eps/\sqrt{8})$.
\end{lemma}
\begin{proof}
    Let $\bm{X}^{(1)},\dots,\bm{X}^{(n)}$ be i.i.d.\ samples from $\mathcal{G}(\mean,\mathbb{I}_d)$. Define \[\bm{Y}^{(k)} := \mathbb{I}(\bm{X}^{(k)} > 0) \in \{0,1\}^d.\] 
    Then the $\bm{Y}^{(k)}$ are i.i.d.\ samples from a binary product distribution $P$ over $\{0,1\}^d$ with mean vector $\bm{p} \in (0,1)^d$. We distinguish the two cases,
    \paragraph{Case 1: $\|\mean\|_2 = 0$.} Then, for all $i$, $p_i =\Pr[\bm{Y}^{(k)}_i = 1] = 1/2$. Thus, $\bm{p} = \bm{u}_d$.
    \paragraph{Case 2: if $\|\mean\|_{2} \ge \eps$.} Then, for all $i$, $p_i =\Pr[\bm{Y}^{(k)}_i =1]=\frac{1}{2}\operatorname{Erfc}\left(-\mu_i/\sqrt{2}\right)$. Using standard properties of Erfc, we have
    \[
    \sum_{i=1}^d(p_i-1/2)^2 \ge \frac{1}{8}\min\left(\sum_{i=1}^d\mu_i^2,d/2 \right) = \frac{1}{8}\min\left(\|\mean\|_2^2,d/2 \right) \ge \frac{\eps^2}{8}
    \]
    Thus, $\|\bm{p}-\bm{u}_d\|_2 \ge \eps/\sqrt{8}$.

    Hence, to distinguish between the two cases, it is sufficient to test if $\bm{p} = \bm{u}_d$ versus $\|\bm{p}-\bm{u}_d\|_2 \ge \eps/\sqrt{8}$, which is exactly the binary product mean testing problem. 
\end{proof}

\begin{theorem}[\cite{AcharyaCT20} \label{private_coin} Theorem $1$] For all $\ell \in [d]$, there exists a private-coin protocol for distributed Gaussian mean testing under $\ell$-bit communication constraint with $O(d^{3/2}/\ell\eps^2)$ users.
\end{theorem}
\begin{proof}
    Let $\bm{X}^{(1)},\dots,\bm{X}^{(n)}$ be i.i.d.\ samples from $\mathcal{G}(\mean,\mathbb{I}_d)$ where user $k$ observes $\bm{X}^{(k)}$. Define $\bm{Y}^{(k)} := \mathbf{1}(\bm{X}^{(k)} > 0) \in \{0,1\}^d$. Then $\bm{Y}^{(k)}$ are i.i.d.\ samples from a binary product distribution $P$ over $\{0,1\}^d$ with mean vector $\bm{p} \in (0,1)^d$. First, note that if $\ell = d$, each user can transmit the entire vector $\bm{Y}^{(k)}$, yielding a protocol with $\Theta(\sqrt{d}/\eps^2)$ users using~\cref{lemma:reduction_to_bpmt}.

    When $\ell < d$, we use the ``\emph{simulate and infer}'' technique of \cite{ACT:IT2}: $d/\ell$ users can simulate \emph{one} sample of the $d$-dimensional product distribution, by having send (a different set of) $\ell$ coordinates. Combining this with the $\Theta(\sqrt{d}/\eps^2)$ centralized sample complexity of binary product mean testing, this gives  a protocol with sample complexity
    \[
    \frac{d}{\ell}\cdot O\left(\frac{\sqrt{d}}{\eps^2}\right) = O\left(\frac{d^{3/2}}{\ell\eps^2}\right)
    \]
    as claimed.
\end{proof}

\section{Blockwise Randomized Hadamard Transform}
    \label{sec:brht}
In this section, we define and analyze the properties of a key primitive we introduce, the \emph{Blockwise Randomized Hadamard Transform (BRHT)}. At a high level, the idea is to replace a uniformly random rotation of $\mathbb{R}^d$ by a discrete analogue, which can be sampled using a \emph{finite} number of random bits. An obvious candidate would be to use either a matrix of uniformly random $\pm 1$; however, our aim is to preserve the $\ell_2$ geometry, and this fails to do so. Indeed, we want an isometry, something which would play the same role as a uniformly random rotation, and in particular a transformation $R$ which, when applied to a spherical Gaussian, would still result in a spherical Gaussian: which requires $R^\top R= \mathbb{I}_d$. This makes using a Hadamard matrix a very natural idea: however, with such a deterministic construction, we lack the randomness which enables more efficient testing protocols. To remedy this while still using a minimum amount of randomness, we consider a structured random orthogonal transform derived from the Hadamard matrix, which we term Blockwise Randomized Hadamard Transform (BRHT), and provides strong sketching guarantees with control over randomness requirement. We formally define the BRHT below.
\begin{definition}[Blockwise Randomized Hadamard Transform (BRHT)]\label{def:BRHT}
Let $d,\tilded \in \mathbb{N}$ satisfy $\tilded \mid d$, and define $b := d/\tilded$. A \emph{$(d,\tilded)$-Blockwise Randomized Hadamard Transform (BRHT)} is the random matrix
\[R := H_dD\]
where $H_d \in \mathbb{R}^{d\times d}$ is a normalized Sylvester Hadamard matrix and $D \in \mathbb{R}^{d\times d}$ is a blockwise diagonal matrix:
\[
    \operatorname{diag}(D) = \left( \bm{\delta}^{(1)},\bm{\delta}^{(2)},\cdots,\bm{\delta}^{(b)} \right)
\]
where for every $i \in [b]$, $\bm{\delta}^{(i)} = (\delta_i,\delta_i,\cdots,\delta_i) \in \{-1,1\}^{\tilded}$ and $\delta_1,\dots, \delta_b$ are $4$-wise independent Rademacher random variables.
\end{definition}

\begin{lemma}[Orthogonality]
Let $R$ be a $(d,\tilded)$-BRHT. Then $RR^\top = \mathbb{I}_d.$
\end{lemma}

\begin{proof}
Since $H_d$ is orthogonal and $D$ is diagonal with entries in $\{\pm1\}$, $RR^\top = (H_dD)(H_dD)^\top = H_dDD^\top H_d^\top = H_dH_d^\top = \mathbb{I}_d.$
\end{proof}

\noindent With this in hand, we are ready to state and prove the key property of BRHT our algorithms will hinge upon.
\begin{lemma}[$(d,\tilded)$-compression]
\label{lemma:compression} 
Let $R$ be a $(d,\tilded)$-BRHT. Then, for every $\mean \in \mathbb{R}^d$ and every $t \ge 1$,
\[
\Pr\!\left[\left\|(R\mean)_{[1:t\tilded]} \right\|_2^2 \ge \frac1{100}\cdot\frac{t\tilded}{d}\cdot \|\mean\|_2^2\right] \ge \frac{8}{25},
\]
where the probability is over the choice of $R$.
\end{lemma}

\begin{proof}
Let $b := d/\tilded$ and $\rho := \|\mean\|_2.$ By~\cref{def:BRHT},
\[R = H_dD. \]
Define
\[
Z := \left\|(R\mean)_{[1:t\tilded]}\right\|_2^2.
\]
Partition
\[
\mean = (\mean_1,\dots,\mean_b), \qquad \mean_i \in \mathbb{R}^{\tilded}.
\]
Using the Kronecker decomposition
\[
H_d = H_b \otimes H_{\tilded},
\]
the $r$-th block of $R\mean$ satisfies
\[
(R\mean)_{[(r-1)\tilded+1:r\tilded]} = \sum_{i=1}^b (H_b)_{r,i}\delta_i H_{\tilded}\mean_i. \]
Define
\[
Z_r := \left\|(R\mean)_{[(r-1)\tilded+1:r\tilded]}\right\|_2^2,
\]
so that
\[
Z = \sum_{r=1}^t Z_r.
\]
We first compute $\mathbb{E}[Z]$. Expanding $Z_r$ gives
\[
\begin{aligned}
Z_r &= \left\|\sum_{i=1}^b(H_b)_{r,i}\delta_i H_{\tilded}\mean_i \right\|_2^2\\
&= \sum_{i_1,i_2 \in [b]}(H_b)_{r,i_1}(H_b)_{r,i_2} \delta_{i_1}\delta_{i_2} \left\langle H_{\tilded}\mean_{i_1}, H_{\tilded}\mean_{i_2} \right\rangle.
\end{aligned}
\]
Since $H_{\tilded}$ is orthogonal,
\[
\left\langle H_{\tilded}\mean_{i_1}, H_{\tilded}\mean_{i_2} \right\rangle = \langle \mean_{i_1},\mean_{i_2}\rangle.
\]
Therefore,
\[
\begin{aligned}
Z_r = \sum_{i=1}^b (H_b)_{r,i}^2 \|\mean_i\|_2^2 + \sum_{i_1\neq i_2} (H_b)_{r,i_1}(H_b)_{r,i_2} \delta_{i_1}\delta_{i_2}\langle \mean_{i_1},\mean_{i_2}\rangle.
\end{aligned}
\]
Since $(H_b)_{r,i}^2 = 1/b$,
\[
Z_r = \frac{\rho^2}{b} + \sum_{i_1\neq i_2} (H_b)_{r,i_1}(H_b)_{r,i_2} \delta_{i_1}\delta_{i_2} \langle \mean_{i_1},\mean_{i_2}\rangle.
\]
By pairwise independence and symmetry of the Rademacher variables,
\[
\mathbb{E}[\delta_{i_1}\delta_{i_2}] = 0 \qquad (i_1\neq i_2),
\]
and hence
\[
\mathbb{E}[Z_r] = \frac{\rho^2}{b}.
\]
Summing over $r \in [t]$,
\[
\mathbb{E}[Z] = \frac{t\rho^2}{b}.
\]
We now bound $\mathbb{E}[Z^2]$. For $r_1,r_2 \in [t]$,
\[
\begin{aligned}
\mathbb{E}[Z_{r_1}Z_{r_2}] = \frac{\rho^4}{b^2} + \sum_{\substack{i_1\neq i_2\\ j_1\neq j_2}} (H_b)_{r_1,i_1} (H_b)_{r_1,i_2} (H_b)_{r_2,j_1} (H_b)_{r_2,j_2} \cdot \mathbb{E}[\delta_{i_1}\delta_{i_2}\delta_{j_1}\delta_{j_2}] \langle \mean_{i_1},\mean_{i_2}\rangle \langle \mean_{j_1},\mean_{j_2}\rangle.
\end{aligned}
\]
By $4$-wise independence and symmetry of the Rademacher variables, only terms in which every $\delta_i$ appears with even multiplicity contribute to the expectation. Thus,
\[
\begin{aligned}
\mathbb{E}[Z_{r_1}Z_{r_2}] = \frac{\rho^4}{b^2} + 2
\sum_{i_1\neq i_2}(H_b)_{r_1,i_1}(H_b)_{r_1,i_2}(H_b)_{r_2,i_1}(H_b)_{r_2,i_2}\langle \mean_{i_1},\mean_{i_2}\rangle^2.
\end{aligned}
\]
Using
\[
(H_b)_{r,i} = \pm \frac1{\sqrt b}
\]
and Cauchy--Schwarz,
\[
\begin{aligned}
\mathbb{E}[Z_{r_1}Z_{r_2}] &\le \frac{\rho^4}{b^2} + \frac{2}{b^2}\sum_{i_1\neq i_2} \|\mean_{i_1}\|_2^2 \|\mean_{i_2}\|_2^2\\ 
&\le \frac{3\rho^4}{b^2},
\end{aligned}
\]
where we used
\[
\sum_{i=1}^b\|\mean_i\|_2^2 = \rho^2.
\]
Summing over $r_1,r_2 \in [t]$ yields
\[
\mathbb{E}[Z^2] \le \frac{3t^2\rho^4}{b^2}.
\]
Applying the Paley--Zygmund (\cref{lemma:paley-zygmund}),
\[
\Pr\!\left[Z > \theta \mathbb{E}[Z]\right] \ge (1-\theta)^2 \frac{\mathbb{E}[Z]^2}{\mathbb{E}[Z^2]} \ge \frac{(1-\theta)^2}{3}.
\]
Choosing $\theta = 1/100$ gives
\[
\Pr\!\left[Z > \frac{t\rho^2}{100b}\right] \ge \frac{8}{25}.
\]
Since
\[
b = \frac{d}{\tilded},
\]
the claim follows.
\end{proof}

\noindent To conclude, we analyze the amount of randomness required to sample a BRHT.
\begin{remark} \label{brht_randomness}
    A $(d, \tilded)$-BRHT requires at most $4\log(d/\tilded)$ bits of randomness as $(d/\tilded)$ $4$-wise Rademacher random variables can be generated using a standard construction involving a degree-$3$ polynomial over an appropriate field (see, e.g.,~\cite[Corollary 3.34]{Vadhan12}).
\end{remark}

\begin{remark}[Computational complexity] For $\bm{v} \in \mathbb{R}^d$, $H_d\bm{v}$ can be computed in $O(d\log d)$ time via the Fast Walsh--Hadamard transform. This computation dominates the running time of our protocols.
\end{remark}

\section{Limited randomness} 
In this section, we prove \cref{sdgmt}, by providing and analyzing an explicit protocol for the problem. This protocol is quite simple: it first uses the shared randomness to jointly sample a BRHT $R$ with suitable parameters; each user then applies $R$ to their samples to obtain an instance of Gaussian mean testing in a lower-dimensional space. Having used all the shared randomness available, they then run the private-coin protocol of~\cref{private_coin} on this lower-dimensional problem.

\begin{theorem}[\cref{sdgmt}, restated]
For $\ell \in [d]$ and $s\ge 0$, there exists an $s$-public coin protocol for distributed Gaussian mean testing under $\ell$-bit communication constraint with 
    \[O\left(\frac{d}{\ell\eps^2}\cdot\sqrt{\max\left(\frac{d}{2^{ \Theta(s) }} , \ell\right)}\right)
    \]
    users, each holding one sample.
\end{theorem}

\begin{proof}
The users use $s/7$ bits of shared randomness to sample a $(d,d_s)$-BRHT $R$, where
\[
d_s := d/2^{\lfloor s/28\rfloor}\,.
\]
    This is possible by~\cref{brht_randomness}, as sampling $R$ requires at most $4\log(d/d_s)$ random bits. (We will explain the reason for the factor $7$ in $s/7$ momentarily.)

    Let $\mean \in \mathbb{R}^d$ be an unknown mean vector and for all $k\in[n],\bm{X}^{(k)} \sim \mathcal{G}(\mean,\mathbb{I}_d)$ be the sample observed by user $k$. Define $\tilded:= d_s\vee\ell$ and $\tilde{\mean} = (R\mean)_{[1:\tilded]}$. Each user applies $R$ to its sample to get $\tilde{\bm{X}}^{(k)} := (R\bm{X}^{(k)})_{[1: \tilded]}$, so that $\tilde{\bm{X}}^{(k)}\sim\mathcal{G}\left(\tilde{\mean},\mathbb{I}_{\tilded}\right)$. By~\cref{lemma:compression}, with probability at least $8/25$,
    \[
        \left\|\tilde{\mean}\right\|_2^2 \ge \frac{1}{100} \cdot \frac {\tilded}{d}\cdot \|\mean\|^2_2
    \]
    Thus, if $\mean = \bm{0}$, $\tilde{\mean} = \bm{0}$ and otherwise $\|\mean\|_2 \ge \eps$ and with probability at least $8/25$, $\|\tilde{\mean}\|_2^2\ge \frac{\tilded\eps^2}{100d}$. The users and referee then use the {private-coin} protocol of~\cref{private_coin} with $\tilde{\bm{X}}^{(k)}$ to distinguish between the two cases with
    \[
        O\left(\frac{(\tilded)^{3/2}}{\ell\left(\frac{\tilded}{d}\eps^2\right)}\right) = O\left( \frac{d\sqrt{\tilded}}{\ell\eps^2} \right)
    \]
    users, (almost) proving the theorem. It remains to amplify the success probability.

    Thus, with $O\left( \frac{d\sqrt{\tilded}}{\ell\eps^2} \right)$ users, if $\mean = \bm{0}$, the referee accepts with probability at least $99/100$ (which is the failure probability from \cref{private_coin}). But when $\|\mean\|_2 \ge \eps$, all we can say is that, by a union bound, the referee rejects with probability at least $1 - (\frac{17}{25} + \frac{8}{25}\cdot\frac{1}{100}) > 3/10$, as now the failure probability comes both from the choice of $R$ and that of \cref{private_coin}. Repeating the protocol $7$ independent times, each time with new $R$, and accepting only if all repetitions output \textsf{accept} reduces the error probability below $1/10$. This increases the number of users only by a constant factor. Since $\tilded = \max(d/2^{\lfloor s/28 \rfloor}, \ell)$, the theorem follows.
\end{proof}
    With $s = O(\log(d/\ell))$, \cref{sdgmt} yields sample complexity $O(d/\sqrt{\ell}\eps^2)$, which matches the public-coin sample complexity in \cite{AcharyaCT20} when users must share an infinite amount of randomness.
    \begin{corollary}[Randomness-efficient optimal public-coin protocol] 
        There exists a protocol for distributed Gaussian mean testing under $\ell$-bit communication constraints with $O({d}/(\sqrt{\ell}\eps^2))$ users, each holding one sample, and sharing $O(\log(d/\ell))$ bits of randomness.
    \end{corollary}
    
\section{Heterogeneous samples}
In this section, we prove~\cref{msampledgmt}. The core idea, which we refer to as \emph{Signal Aggregation}, is to generalize the ``naive strategy'' previously discussed (see~\cref{ft:aggregate}) where all users have the same number $m$ of samples and aggregate them into a single ``aggregated sample'':
\[
\bm{X}^{(k,1)},\bm{X}^{(k,2)},\cdots, \bm{X}^{(k,m_k)}\sim\mathcal{G}(\mean,\mathbb{I}_d) \quad\leadsto\quad \tilde{\bm{X}}^{(k)} := \frac{1}{\sqrt{m}}\sum_{i=1}^m\bm{X}^{(k,i)} \sim \mathcal{G}(\sqrt{m}\mean,\mathbb{I}_d)
\]
Doing so effectively amplifies the signal, which means that one can then use a protocol designed for \emph{one} sample per user, but with better distance parameter $\eps' := \sqrt{m}\eps$.

In the heterogeneous sample case, one can attempt to do the same thing, but now (since user $k$ has $m_k$ samples, while user $k'$ has a possible different number of samples $m_{k'}$) the resulting aggregated samples no longer have the same distribution in the \textsf{no}-case (when $\mean\neq \bm{0}$). Dealing with the fact that samples are no longer i.i.d.\ is the key technical difficulty in proving~\cref{msampledgmt}.
\begin{theorem}[\cref{msampledgmt}, restated] \label{msampledgmt:restated}
    For $\ell \in [r,rd]$, $s= \Omega(\log(d/\ell))$, and $\bm{m} = (m_1,\dots, m_n) \in \mathbb{Z}_+^n$ satisfying $m_k \in [r,r\eps^{-2}]$ for all $k \in [n]$, there exists a protocol (~\cref{algorithm:mdgmt_alg}) for distributed Gaussian mean testing under $\ell$-bit communication constraint with $n$ users provided that 
    \[
    \frac{\sum_{1\leq k_1 \neq k_2 \leq n} \sqrt{m_{k_1}{m_{k_2}}}}{n}  \geq C\cdot\frac{d}{\sqrt\ell\eps^2} 
    \]
    where $C,r>0$ are absolute constants.
\end{theorem}

\begin{proof}
We first describe the protocol in~\cref{algorithm:mdgmt_alg} and then analyze its correctness.

\begin{algorithm}
\caption{$\ell$-bit, $\bm{m}$-sample distributed Gaussian mean testing}
\label{algorithm:mdgmt_alg}

\vspace{0.3em}

User $k$, for every $k \in [n]$:
\begin{algorithmic}[1]
\Require
$\forall j\in[m_k]$,
$\bm{X}^{(k,j)}
\sim
\mathcal{G}(\mean,\mathbb{I}_d)$;
communication budget $\ell \in[7,7d]$;
$\forall t\in[7]$,
$R_t$: $(d,7d/\ell)$-BRHT

\For{$t \in [7]$}
    \State Samples $\bm{X}^{(k,t)}\gets\frac1{\sqrt{\lfloor m_k/7 \rfloor}}\sum_{j=(t-1)\lfloor m_k/7\rfloor+1}^{t \lfloor m_k/7 \rfloor}\bm{X}^{(k,j)}$
    
    \State $\tilde{\bm{X}}^{(k,t)} \gets R_t\bm{X}^{(k,t)}$
    
    \State $\bm{y}^{(k,t)} \gets \mathbf{1}\left(\tilde{\bm{X}}^{(k,t)}_{[1: \lfloor \ell/7 \rfloor]} > 0 \right)$ 
\EndFor

\State
\texttt{Send} $\bm{Y}^{(k)}=(\bm{y}^{(k,1)},\dots,\bm{y}^{(k,7)})$
\end{algorithmic}

\vspace{1em}

Referee:
\begin{algorithmic}[1]
\Require
Messages $\forall k\in[n]$,
$\bm{Y}^{(k)}\in\{0,1\}^{\ell}$, threshold $\eps$

\State $N \gets  \sum_{k_1\neq k_2} \sqrt{\lfloor m_{k_1}/7\rfloor \lfloor m_{k_2}/7\rfloor}$

\For{$t\in[7]$}
    \State Run~\cref{uniformity_testing} on
    samples $\bm{y}^{(1,t)},\dots,\bm{y}^{(n,t)}$ with  \\ \qquad threshold $\eps' =\frac{\eps}{80}\sqrt{\frac{\ell N}{7dn(n-1)}}$
    \Comment{$\tau \gets (\eps')^2/2$ }
    
    \If{output is \textsf{reject}}
        \State \Return \textsf{reject}
    \EndIf
\EndFor

\State \Return \textsf{accept}
\end{algorithmic}
\end{algorithm}

The protocol performs $7$ independent repetitions for probability amplification. We analyze a single repetition and omit the superscript $(t)$ throughout. For simplicity, we also ignore flooring operations, as they affect the bounds only up to constant factors. Fix one repetition, and let $R$ denote the corresponding $(d,7d/\ell)$-BRHT. Define
\[
\tilde{\mean} := (R\mean)_{[1:\ell/7]}.
\]
For every user $k$,
\[
\bm{X}^{(k)} = \frac1{\sqrt{m_k/7}}\sum_{j=1}^{m_k/7} \bm{X}^{(k,j)} \sim \mathcal{G} \!\left(\sqrt{m_k/7}\,\mean, \mathbb{I}_d \right).
\]
Since $R$ is orthogonal,
\[
\tilde{\bm{X}}^{(k)} := (R\bm{X}^{(k)})_{[1:\ell/7]} \sim \mathcal{G}\!\left(\sqrt{m_k/7}\,\tilde{\mean}, \mathbb{I}_{\ell/7}\right).
\]
Define
\[
\bm{y}^{(k)} = \mathbf{1} (\tilde{\bm{X}}^{(k)} > 0) \in \{0,1\}^{\ell/7}.
\]
The vectors $\bm{y}^{(1)},\dots,\bm{y}^{(n)}$ are independent, though generally not identically distributed. For every coordinate $i\in[\ell/7]$, define
\[
p_i^{(k)} := \Pr[y_i^{(k)} = 1], \qquad \tilde p_i^{(k)} := p_i^{(k)} - \frac12.
\]
The referee computes
\[
T_i := \sum_{k_1\neq k_2}(y_i^{(k_1)}-1/2)(y_i^{(k_2)}-1/2).
\]
Since for all $k_1 \neq k_2$, $\tilde p_i^{(k_1)}\tilde p_i^{(k_2)} \ge 0$, applying~\cref{lemma:testing_lemma} coordinate-wise yields
\[
\mathbb{E}[T_i] = \sum_{k_1\neq k_2} \tilde p_i^{(k_1)} \tilde p_i^{(k_2)}, 
\]
and
\[
\operatorname{Var}(T_i) \le \frac{n(n-1)}8 + (n-2)\mathbb{E}[T_i].
\]
Define
\[
T := \frac1{n(n-1)} \sum_{i=1}^{\ell/7} T_i.
\]
Then
\[
\begin{aligned}
\mathbb{E}[T] &= \frac1{n(n-1)} \sum_{i=1}^{\ell/7} \sum_{k_1\neq k_2} \tilde p_i^{(k_1)} \tilde p_i^{(k_2)},\\
\operatorname{Var}(T)
&=\frac1{n^2(n-1)^2} \sum_{i=1}^{\ell/7} \operatorname{Var}(T_i)\\
&\le\frac{\ell}{56\,n(n-1)} + \frac{(n-2)\mathbb{E}[T]}{n(n-1)}.
\end{aligned}
\]
We now analyze the two cases:

\paragraph{Case 1: $\mean=\bm{0}$.} In this case, $p_i^{(k)} = 1/2$ for all $i,k,$ and hence $\tilde p_i^{(k)} = 0.$ Therefore,
\[
\mathbb{E}[T]=0.
\]
Applying Chebyshev's inequality,
\[
\Pr[T\ge\tau] \le \frac{\operatorname{Var}(T)}{\tau^2} \le \frac{\ell} {56\,n(n-1)\tau^2}.
\]
Using the threshold
\[
\tau = \frac12 \cdot\frac1{6400} \cdot \frac{N}{n(n-1)} \cdot \frac{\ell\eps^2}{7d},
\]
we obtain
\[
\Pr[T\ge\tau] = O\!\left(\frac{d^2 n(n-1)}{N^2\ell\eps^4}\right) < \frac1{100},
\]
whenever
\[
\frac{N}{n} > C_1 \cdot \frac{d}{\sqrt{\ell}\eps^2}
\]
for a sufficiently large absolute constant $C_1>0$.

\paragraph{Case 2: $\|\mean\|_2\ge\eps$.} Using standard properties of Erfc,
\[
|\tilde p_i^{(k)}| \ge \frac18\sqrt{\frac{m_k}{7}}\,|\tilde\mu_i|,
\]
Therefore,
\[
\begin{aligned}
\mathbb{E}[T] &= \frac1{n(n-1)} \sum_{i=1}^{\ell/7}\sum_{k_1\neq k_2}\tilde p_i^{(k_1)} \tilde p_i^{(k_2)}
\\
&\ge
\frac1{64\,n(n-1)}\sum_{i=1}^{\ell/7}\sum_{k_1\neq k_2}\sqrt{\frac{m_{k_1}}7}\sqrt{\frac{m_{k_2}}7}\,\tilde\mu_i^2\\
&=
\frac{N}{64\,n(n-1)} \|\tilde{\mean}\|_2^2.
\end{aligned}
\]
By~\cref{lemma:compression}, with probability at least $8/25$,
\[
\|\tilde{\mean}\|_2^2 \ge \frac1{100} \cdot \frac{\ell}{7d} \cdot \eps^2.
\]
Hence, with probability at least $8/25$,
\[
\mathbb{E}[T] \ge \frac1{6400} \cdot \frac{N}{n(n-1)} \cdot \frac{\ell\eps^2}{7d}.
\]
Conditioned on this event, applying Chebyshev's inequality,
\[
\begin{aligned}
\Pr[T\le\tau] &= \Pr\!\left(\mathbb{E}[T]-T \ge \frac{\mathbb{E}[T]}2\right)\\
&\le \frac{4\operatorname{Var}(T)}{\mathbb{E}[T]^2}\\ 
&= O\!\left(\frac{d^2 n^2}{N^2\ell\eps^4} + \frac{nd}{N\ell\eps^2}\right).
\end{aligned}
\]
Thus,
\[
\Pr[T\le\tau] < \frac1{100},
\]
provided
\[
\frac{N}{n} > C_2 \cdot \frac{d}{\sqrt{\ell}\eps^2}
\]
for a sufficiently large absolute constant $C_2>0$.

Therefore, a single repetition accepts with probability at least $99/100$, if $\|\mean\|_2 = 0$ and rejects with probability at least $1 - \left(\frac{17}{25} + \frac{8}{25}\cdot\frac1{100} \right) > \frac3{10}$, if $\|\mean\|_2\ge \eps$. Repeating independently $7$ times and accepting only if all repetitions output \textsf{accept} reduces the error probability below $1/10$. Finally, by~\cref{brht_randomness}, the $7$ independent BRHTs require only $O(\log(d/\ell))$ shared random bits.
\end{proof}

\begin{remark}
Note that for a fixed sampling budget $m = \sum_{k \in [n]}m_k$, $\sum_{k_1\neq k_2 \in[n]}\sqrt{m_{k_1}}{\sqrt{m_{k_2}}}$ is maximized when $m_{k_1}=m_{k_2}=m/n$ for all $k_1,k_2 \in [n]$. Then, we need $m = O(d/\sqrt{\ell}\eps^2)$.    
\end{remark}

\section{Heterogeneous communication}
Having addressed limited (public) randomness and heterogeneous sample sizes in the previous two sections, we now turn to the last generalization considered in this work, heterogeneous communication budgets. Specifically, we will establish~\cref{combudgetsdgmt}, restated below:

\begin{theorem}[\cref{combudgetsdgmt}, restated] \label{combudgetsdgmt:restated}
    For $\bm{\ell} = (\ell_1,\dots, \ell_n) \in \mathbb{Z}_+^n$ and $s= \Omega(\log(d/\|\bm{\ell}\|_\infty))$, there exists a protocol for distributed Gaussian mean testing under $\bm{\ell}$-bit communication constraint with $n$ users, where user $k$ has one sample and can send $\ell_k$ bits to the referee, provided that
    \[
    \frac{\|\bm{\ell}\|_1}{\sqrt{\|\bm{\ell}\|_\infty}}  \geq C\cdot\frac{d}{\eps^2} 
    \]
    where $C>0$ is an absolute constant.
\end{theorem}
\begin{proof}
    In view of the previous protocols, the idea is relatively straightforward: the users use $s/7$ bits of shared randomness to sample a $(d,d_s)$-BRHT $R$, for
\[
d_s := d/2^{\lfloor s/28\rfloor}\,.
\]
Let $\tilded:= d_s\vee\max_{1\leq k\leq n}\ell_k = \Theta(\max_{1\leq k\leq n}\ell_k)$ (the last equality by our assumption on $s$) and $\tilde{\mean} = (R\mean)_{[1:\tilded]}$.
As in the proof of~\cref{sdgmt}, each user applies $R$ to its sample to get $\tilde{\bm{X}}^{(k)} = (R\bm{X}^{(k)})_{[1: \tilded]}$, so that $\tilde{\bm{X}}^{(k)}\sim\mathcal{G}\left(\tilde{\mean},\mathbb{I}_{\tilded}\right)$. Using~\cref{lemma:compression}, with probability at least $8/25$,
    \[
        \left\|\tilde{\mean}\right\|_2^2 \ge \frac{1}{100} \cdot \frac {\tilded}{d}\cdot \|\mean\|^2_2
    \]
so that in one case $\tilde{\mean} = \bm{0}$ and in the other $\|\tilde{\mean}\|_2^2 =\Omega(\frac{\tilded\eps^2}{d})$ with constant probability. This is where we depart from \cref{sdgmt}. The users and referee will now use the same idea of distributed simulation, but in a slightly different fashion: recall that all parameters are known to all parties, and that by choice of $\tilded$ we have $\ell_k \leq \tilded$ for all $k$. Then, user 1 sends the first $\ell_1$ bits of their resulting $\tilded$-dimensional sample; user $2$ sends the following $\ell_2$ bits, wrapping around (i.e., the bits $\ell_1+1, \ell_1+2, \dots, \ell_1+\ell_2 \bmod \tilded$), and so on. Overall, the $\ell_1+\dots+\ell_n = \|\bm{\ell}\|_1$ bits sent allow the referee to reconstruct exactly $n':=\lfloor\|\bm{\ell}\|_1/\tilded\rfloor$ binary samples $\bm{Y}_1,\dots, \bm{Y}_{n'}$ as in the proof of~\cref{private_coin}, which leads to a successful protocol whenever
\[
n' \geq C\cdot \frac{\sqrt{d_s}}{\eps^2(d_s/d)}
\]
i.e.,
\[
\frac{\|\bm{\ell}\|_1}{\|\bm{\ell}\|_\infty} \geq C\cdot \frac{d}{\eps^2\sqrt{\|\bm{\ell}\|_\infty}}
\]
from our settings of $\tilded$ and $d_s$ (after, as before, including the probability amplification over $7$ independent repetitions).
\end{proof}

\section{Putting it all together: \cref{theorem:mix-and-match}}
\noindent With protocols for the three settings, we conclude by providing a protocol for the general setting where shared randomness, sample counts and communication budgets are all free parameters. The protocol is relatively straightforward as it combines the protocols for the three settings.
\begin{theorem}[\cref{theorem:mix-and-match}, restated]\label{theorem:mix-and-match_restated}
    For $s\geq 0$, $\bm{m} = (m_1,\dots, m_n) \in \mathbb{Z}_+^n$, and $\bm{\ell} = (\ell_1,\dots, \ell_n) \in \mathbb{Z}_+^n$, there exists a protocol for distributed Gaussian mean testing with $n$ users sharing $s$ bits of public randomness, where user $k$ has $m_k$ samples and can send $\ell_k$ bits to the referee, provided that there 
    exists a partition $\mathcal{P} = \{P_1,\dots,P_K\}$ of the $n$ users such that
    \[
    \frac{ \sum_{j_1\neq j_2 \in [K]}\sqrt{\min_{i_1\in P_{j_1}}(m_{i_1}) \cdot \min_{i_2\in P_{j_2}}(m_{i_2})}}{K} \ge C \cdot \frac{d}{\eps^2\sqrt{\max\left(\frac{d}{2^{\Theta(s)}} ,\|\bm{\ell}\|_\infty \right)}}
    \]
    and, for every $P \in \mathcal{P}$, $\sum_{i\in P}\ell_i \ge r\cdot\max\left(\frac{d}{2^{\Theta(s)}} ,\|\bm{\ell}\|_\infty \right)$ and $\min_{i \in P}(m_i) \in [r,r\eps^{-2}]$;
    where $C,r>0$ are absolute constants.
\end{theorem}
\begin{proof} 
 Let $\mathcal{P} = \{P_1,\dots,P_K\}$ be a partition of users (for all $j_1,j_2$ $P_{j_1} \cap P_{j_2} = \emptyset$, $\bigcup_{j}P_j = [n]$) such that for all $j$, $\sum_{i\in P_j}\ell_i \ge 7\cdot\max\left(\frac{d}{2^{\lfloor s/28\rfloor}} ,\|\bm{\ell}\|_\infty \right)$. Define 
 \[
    \tilded := \max\left(\frac{d}{2^{\lfloor s/28\rfloor}},\|\bm{\ell}\|_\infty\right)
 \]
 Using the distributed simulation idea in \cref{combudgetsdgmt:restated}, users in each $P_j \in \mathcal{P}$ can simulate \emph{one} $\tilded$-dimensional binary sample. But, users in $P_j$ hold $(m_i)_{i\in P_j}$ samples. Thus, the setting can be reduced to $K = |\mathcal{P}|$ users holding $(\min_{i \in P_1}(m_i),\dots,\min_{i \in P_K}(m_i))$ samples under $\tilded$-bit communication constraints, for which we invoke the protocol of~\cref{algorithm:mdgmt_alg}.
\end{proof}

\bibliographystyle{alpha}
\bibliography{ref}

\appendix

\section{Deferred Proofs}

\begin{lemma} \label{lemma:testing_lemma}
Let $Z_1,\dots,Z_n$ be independent random variables such that $Z_k \in \{-1/2,1/2\}$ and $\mathbb E[Z_k]=p_k$. Assume furthermore $p_{k_1}p_{k_2} \ge 0$ for all $k_1\neq k_2$ and define
\[
T := \sum_{k_1\neq k_2} Z_{k_1}Z_{k_2}.
\]
Then the following holds:
\begin{enumerate}
    \item $\mathbb E[T] = \sum_{k_1\neq k_2}p_{k_1}p_{k_2}$
    \item $\operatorname{Var}(T) \le \frac{n(n-1)}8 + (n-2)\mathbb{E}[T]$
\end{enumerate}
\end{lemma}

\begin{proof}
Define
\[
S := \sum_{k\in[n]} Z_k.
\]
Since $Z_k^2 = 1/4$, we have
\[
T = \sum_{k_1\neq k_2} Z_{k_1}Z_{k_2} = \left(\sum_{k\in[n]} Z_k\right)^2 - \sum_{k\in[n]} Z_k^2 = S^2-\frac n4.
\]

\paragraph{1. $\mathbb{E}[T]$.}
Using independence,
\[
\mathbb E[S] = \sum_{k\in[n]} p_k
\]
and
\[
\operatorname{Var}(S) = \sum_{k\in[n]}\left(\frac14-p_k^2\right).
\]
Hence,
\[
\begin{aligned}
\mathbb E[S^2] &=\operatorname{Var}(S)+\mathbb E[S]^2 \\
&= \sum_{k\in[n]} \left(\frac14-p_k^2\right) + \left(\sum_{k\in[n]} p_k\right)^2 \\
&= \frac n4 + \sum_{k_1\neq k_2} p_{k_1}p_{k_2}.
\end{aligned}
\]
Therefore,
\[
\mathbb E[T] = \mathbb E[S^2]-\frac n4 = \sum_{k_1\neq k_2} p_{k_1}p_{k_2}.
\]

\paragraph{2. $\operatorname{Var}[T]$.}
Since $T=S^2-n/4$,
\[
\operatorname{Var}(T) = \operatorname{Var}(S^2) = \mathbb E[S^4]-\mathbb E[S^2]^2.
\]
Expanding $S^4$ and using independence together with
\[
\mathbb E[Z_k]=p_k, \qquad Z_k^2=\frac14, \qquad \mathbb E[Z_k^3]=\frac{p_k}{4}, \qquad Z_k^4=\frac1{16},
\]
we obtain
\[
\begin{aligned}
\mathbb E[S^4] &=\frac n{16}+ \frac{3n(n-1)}{16} \\
&\quad + \sum_{k_1\neq k_2} p_{k_1}p_{k_2} + \frac{3(n-2)}2 \sum_{k_1\neq k_2} p_{k_1}p_{k_2} \\
&\quad + \sum_{k_1\neq k_2\neq k_3\neq k_4} p_{k_1}p_{k_2}p_{k_3}p_{k_4}.
\end{aligned}
\]
On the other hand,
\[
\mathbb E[S^2]^2 = \frac{n^2}{16} + \frac n2 \sum_{k_1\neq k_2} p_{k_1}p_{k_2} + \left(\sum_{k_1\neq k_2} p_{k_1}p_{k_2}\right)^2.
\]
Subtracting,
\[
\operatorname{Var}(T)=\mathbb E[S^4]-\mathbb E[S^2]^2
=
\frac{n(n-1)}8 + (n-2)\sum_{k_1\neq k_2}p_{k_1}p_{k_2} - 4\sum_{k_1\neq k_2 \neq k_3}p^2_{k_1}p_{k_2}p_{k_3} -2\sum_{k_1\neq k_2}p_{k_1}^2p_{k_2}^2.
\]
Since $p_{k_1}p_{k_2} \ge 0$ for all $k_1\neq k_2$ and
\[
\mathbb E[T] = \sum_{k_1\neq k_2} p_{k_1}p_{k_2},
\]
the result follows.
\end{proof}

\begin{theorem} \label{proof_uniformity_testing}~\cref{uniformity_testing} achieves sample complexity $\Theta(\sqrt d/\eps^2)$ for testing whether $\bm p = \bm u_d$ versus $\|\bm p-\bm u_d\|_2 \ge \eps$ for a $d$-dimensional binary product distribution with mean vector $\bm p$.
\end{theorem}

\begin{proof}
Let $\tilde{\bm p}:=\bm p-\bm u_d.$ For each coordinate $i\in[d]$, define
\[
T_i:=\sum_{k_1\neq k_2\in[n]}\left(X_i^{(k_1)}-\frac12\right)\left(X_i^{(k_2)}-\frac12\right),
\]
and
\[
T := \frac{1}{n(n-1)}\sum_{i=1}^d T_i.
\]
Fix $i\in[d]$, and define
\[
Z_k:=X_i^{(k)}-\frac12.
\]
Then $Z_1,\dots,Z_n$ are independent random variables taking values in $\{-1/2,1/2\}$, with $\mathbb E[Z_k] = \tilde p_i.$
Applying~\cref{lemma:testing_lemma},
\[
\mathbb E[T_i] = n(n-1)\tilde p_i^2
\]
and
\[
\operatorname{Var}(T_i)\le \frac{n(n-1)}8 + n(n-1)(n-2)\tilde p_i^2.
\]
Summing over coordinates,
\[
\mathbb E[T] = \|\tilde{\bm p}\|_2^2 = \|\bm p-\bm u_d\|_2^2,
\]
and
\[
\begin{aligned}
\operatorname{Var}(T) &= \frac{1}{n^2(n-1)^2} \sum_{i=1}^d \operatorname{Var}(T_i) \\
&\le \frac{d}{8n(n-1)} + \frac{n-2}{n(n-1)}\|\tilde{\bm p}\|_2^2.
\end{aligned}
\]
We now analyze the two cases:
\paragraph{Case 1: $\bm p=\bm u_d$.}
In this case, $\mathbb E[T]=0$ and therefore
\[
\operatorname{Var}(T)
\le
\frac{d}{8n(n-1)}.
\]
By Chebyshev's inequality,
\[
\Pr\!\left(T>\frac{\eps^2}{2}\right)
\le
\frac{4\operatorname{Var}(T)}{\eps^4}
\le
O\!\left(\frac{d}{n^2\eps^4}\right).
\]
\paragraph{Case 2: $\|\bm p-\bm u_d\|_2\ge\eps$.}
In this case,
\[
\mathbb E[T]
=
\|\bm p-\bm u_d\|_2^2
\ge
\eps^2.
\]
Again by Chebyshev's inequality,
\[
\begin{aligned}
\Pr\!\left(T\le\frac{\eps^2}{2}\right)
&= \Pr\!\left(\mathbb E[T]-T \ge \frac{\mathbb E[T]}{2} \right) \\
&\le \frac{4\operatorname{Var}(T)}{\mathbb E[T]^2} \\
&\le O\!\left(\frac{d}{n^2\eps^4} + \frac{1}{n\eps^2} \right).
\end{aligned}
\]
Thus, choosing
\[
n > C\cdot(\sqrt d/\eps^2)
\]
for some absolute constant $C > 0$, makes both error probabilities smaller than $1/100$ .
\end{proof}
\end{document}